\documentclass[letterpaper,10 pt,conference]{ieeeconf}

\IEEEoverridecommandlockouts                              

\usepackage{array}
\usepackage{bm}
\usepackage{url}
\usepackage{amsmath}
\usepackage{amssymb}
\usepackage[backend=biber]{biblatex}
\usepackage{todonotes}
\usepackage{glossaries}
\usepackage[capitalize]{cleveref}
\usepackage{pgfplots}
\usepackage{tikz}
\usepackage{algorithm2e}
\usepackage{booktabs,array}
\pgfplotsset{compat=1.13}
\newtheorem{theorem}{Theorem}
\newtheorem{definition}[theorem]{Definition}

\newtheorem{corollary}[theorem]{Corollary}
\newtheorem{remark}[theorem]{Remark}

\newcommand{\x}{\bm{x}}
\newcommand{\y}{\bm{y}}
\newcommand{\z}{\bm{z}}
\DeclareMathOperator{\Span}{span}
\DeclareMathOperator{\id}{id}
\DeclareMathOperator{\rank}{rank}

\newacronym{desy}{DESY}{Deutsches Elektronen-Synchrotron}
\newacronym{euxfel}{European XFEL}{European X-Ray Free-Electron Laser}
\newacronym{xfel}{XFEL}{X-Ray Free-Electron Laser}
\newacronym{lbsync}{LbSync}{laser-based optical synchronization}
\newacronym[longplural={linear matrix inequalities}]{lmi}{LMI}{linear matrix inequality}
\newacronym{lo}{LO}{local oscillator}
\newacronym{lti}{LTI}{linear time-invariant}
\newacronym{mlo}{MLO}{master laser oscillator}
\newacronym{mo}{MO}{master timing reference oscillator}
\newacronym{pi}{PI}{proportional-integral}
\newacronym{pll}{PLL}{phase-locked loop}
\newacronym{ppl}{PPL}{pump-probe laser}
\newacronym{rf}{RF}{radio-frequency}
\newacronym{rms}{RMS}{root-mean-square}
\newacronym{slo}{SLO}{subsidiary laser oscillator}
\newacronym{snr}{SNR}{signal-to-noise ratio}
\newacronym{vco}{VCO}{voltage controller oscillator}
\newacronym{inr}{INR}{improvement-to-noise ratio}
\newacronym{bo}{BO}{Bayesian optimization}
\newacronym[longplural={Gaussian processes}]{gp}{GP}{Gaussian process}
\newacronym{ei}{EI}{expected improvement}
\newacronym{rkhs}{RKHS}{reproducing kernel Hilbert space}
\newacronym{ucb}{UCB}{upper confidence bound}
\newacronym{lcb}{LCB}{lower confidence bound}
\newacronym{pdf}{PDF}{probability density function}
\newacronym{cdf}{CDF}{cumulative distribution function}
\newacronym{smgo}{SMGO-\(\delta\)}{set membership global optimization}
\newacronym{awgn}{AWGN}{additive white Gaussian noise}
\newacronym{psd}{PSD}{power spectral density}
\newacronym{lsu}{LSU}{link stabilization unit}
\newacronym{oxc}{OXC}{optical cross-correlator}
\newacronym{icm}{ICM}{intrinsic co-regionalization model}
\newacronym{lmc}{LMC}{linear model of co-regionalization}
\newacronym{mle}{MLE}{maximum likelihood estimation}
\newacronym{mcmc}{MCMC}{Markov chain Monte Carlo}
\newacronym{map}{MAP}{maximum a posteriori} 
\newacronym{lkj}{LKJ}{Lewandowski-Kurowicka-Joe}
\newacronym{ipm}{IPM}{interval-predictor model}
\newacronym{nvs}{n.v.s.}{normed vector space}
\newacronym{mpc}{MPC}{model predictive control}

\title{\LARGE \bf
Robust Nonlinear System Identification in Reproducing Kernel Hilbert Spaces via Scenario Optimization
}

\author{Jannis Lübsen$^{1}$ and Annika Eichler$^{1,2}$
\thanks{$^{1}$Institute of Control Systems, Hamburg University of Technology, Hamburg, Germany
        {\tt\small \{jannis.luebsen,annika.eichler\}@tuhh.de}}%
\thanks{$^{2}$Machine Beam Controls, Deutsches Elektronen-Synchrotron (DESY), Hamburg, Germany
        {\tt\small annika.eichler@desy.de}}%
}

\begin{document}

\maketitle
\thispagestyle{empty}
\pagestyle{empty}

\begin{abstract}
This paper proposes a method for constructing one-step prediction tubes for nonlinear systems using \acrlongpl{rkhs}. We approximate a bounded \gls{rkhs} hypothesis set by a finite-dimensional subspace using bounds based on $n$-widths and a greedy algorithm for basis reduction. For kernels whose native spaces are norm-equivalent to Sobolev spaces, we derive how the required basis size scales with kernel smoothness and input dimension. This finite-dimensional representation enables the use of convex scenario optimization to obtain violation guarantees for the learned predictor without requiring an a priori bound on the true system’s \gls{rkhs} norm or Lipschitz constant. The method is demonstrated on an obstacle-avoidance task. We also discuss the main limitations of the current analysis, including dimensional scaling and dependence on i.i.d. data.
\end{abstract}

\IEEEpeerreviewmaketitle

\section{Introduction}

Data-driven modeling and control of dynamical systems have attracted significant attention in recent years. In particular, kernel methods have emerged as a powerful tool for learning complex nonlinear relationships from data due to their flexibility and strong theoretical foundations. The use of kernel methods for system identification is not new, and a substantial body of work has developed in this area. The literature can be divided into two sub-branches: one rooted in systems and control theory, e.g., \cite{pillonetto2018,khosravi2023}, and one rooted in machine learning, e.g., \cite{suykens2010}. However, except for the linear case \cite{yin2023}, rigorous uncertainty bounds for identified nonlinear kernel models remain limited.

A related operator-theoretic perspective is provided by the Koopman framework. Koopman operators can be directly approximated in an \gls{rkhs} using methods such as extended dynamic mode decomposition~\cite{Williams2015b}. More recently, pointwise error bounds for the Koopman operator in an \gls{rkhs} that is norm-equivalent to a Sobolev space have been derived~\cite{Kohne2024}. The practical implementation of this method can quickly become infeasible as the state space dimension increases due to the curse of dimensionality, a phenomenon that also appears in our analysis. 

Another prominent approach is to use \glspl{gp} for system identification \cite{kocijan2016,Scampicchio2025}, leveraging their uncertainty quantification to design robust model predictive controllers \cite{Kocijan2004,Bradford2020a}. 
This approach, however, also faces a fundamental limitation, i.e., to derive theoretically sound bounds on the uncertainty, strong assumptions on the system are required, e.g., upper bounds on the \gls{rkhs} norm in the frequentist setting or correctly specified hyperparameters in the Bayesian setting.

Robustly identifying nonlinear systems from data without requiring an a priori upper bound on the true system’s \gls{rkhs} norm or a Lipschitz constant remains a significant challenge. 
In this paper, we address this challenge by combining the scenario approach with the \gls{rkhs} framework. Our main contribution is a systematic procedure based on $n$-widths for constructing an effective finite-dimensional subspace that approximates an infinite-dimensional \gls{rkhs} with arbitrary accuracy. We analyze how this subspace's dimension is influenced by the kernel and the input space. This finite-dimensional representation enables the application of the scenario framework \cite{Campi2018} to derive rigorous robustness guarantees for the identified one-step model without requiring knowledge of the system's \gls{rkhs} norm or its Lipschitz constant. 
 We validate the proposed method by using the identified model to solve an obstacle-avoidance problem in simulation.

 The paper is structured as follows: In \cref{sec:theory}, the necessary theoretical background on \glspl{rkhs} and the scenario approach is provided. In \cref{sec:n_widths}, the concept of $n$-widths is introduced, and bounds on the required dimension of the finite-dimensional subspace are derived. In \cref{sec:application}, a numerical method for estimating the effective dimension is presented, and the approach is applied on a robust system identification problem. Finally, in \cref{sec:conclusion}, the work is summarized.

\section{Theory}
\label{sec:theory}
\subsection{Reproducing Kernel Hilbert Spaces}
We denote by $L^{p}(\Omega,\mu)$, for $1 \leq p < \infty$, the space of equivalence classes $[f]_\sim$ of measurable functions $f : \Omega \to \mathbb{R}$ such that
\[
\|f\|_{p} := \left( \int_{\Omega} |f(x)|^{p} \, d\mu(x) \right)^{1/p} < \infty,
\]
where $[f]_\sim$ is the equivalence class of all functions that agree $\mu$-almost everywhere with $f$. 
For $p = \infty$, we use as usual the supremum norm.

The inclusion operator is defined as
\begin{align*}
    \id \,:\, &H \to L^p_\mu(\Omega)\\
     & f\mapsto [f]_\sim,
\end{align*}
for an arbitrary \gls{nvs} $H$, and is assumed to be continuous.

We define the norm \[\|f\|_{W_m^p} := \left(\sum_{0\leq|\alpha|\leq m} \|D^\alpha f\|_{L^p}^p\right)^{1/p}.\] The Sobolev space $W_p^m(\Omega)$ \cite{Adams2003} is defined as the completion of $C^m(\Omega)$, the space of $m$-times continuously differentiable functions, with respect to the norm $\|\cdot\|_{W_m^p} $. Hence, $W_p^m(\Omega)$ contains functions whose weak derivatives up to order $m$ are in $L^p(\Omega)$.

A \gls{rkhs} $H$ is a Hilbert space of functions $f : \Omega \to \mathbb{R}$ such that the evaluation functional $\delta_x : H \to \mathbb{R}$, defined by $\delta_x(f) = f(x)$, is bounded for all $x \in \Omega$. 
By the Riesz representation theorem, there exists a unique function $k_x \in H$ such that $f(x) = \langle f, k_x \rangle_H$ for all $f \in H$.
The function $k : \Omega \times \Omega \to \mathbb{R}$ defined by $k(x,y) = k_y(x)$ is called the reproducing kernel of $H$. 

\subsection{Scenario Approach and Interval Predictor Models}
We use the scenario framework \cite{Campi2018} to certify probabilistic guarantees on the robustness of the model, given a finite set of samples. We describe the probability space by the triplet $(\Delta,\mathcal{F},\mathbb{P})$. In statistical learning, the true probability measure $\mathbb{P}$ is typically unknown, so we rely on a finite number of samples, $\delta_i$, assumed to be drawn i.i.d. from $\mathbb{P}$. Within this framework, the scenario optimization problems addressed in this manuscript are rewritten in their epigraph form and can be cast as a convex scenario program 
\begin{equation}
    \label{eq:scenario_optimization}
\begin{aligned}
    \min_{\theta}\quad &\bm{c}^\top \theta\\
    \text{s.t. }\quad & g(\theta,\delta_i)\le 0,\ i=1,\dots,N, 
\end{aligned}
\end{equation}
where $g(\cdot,\delta)$ is convex in $\theta$.

$\theta\in\mathbb{R}^{n+1}$ is the decision variable, and the $\delta_i$ are the samples. We denote the optimized decision variable as $\theta^\star$.
The violation probability $V(\theta)$ is defined as
\begin{align*}
    V(\theta):=P\{\delta\in\Delta: g(\theta,\delta)>0\},
\end{align*} 
which is the probability that the variable $\theta$ does not satisfy the constraints for random samples $\delta\in\Delta$. Usually we accept a risk $\varepsilon\in (0,1]$ of violating the feasible set defined in \eqref{eq:scenario_optimization}. For a convex scenario program with decision dimension $n+1$, i.i.d. samples, and a unique optimizer, the solution \(\theta^\star\) satisfies 
\begin{align}
    \label{eq:scenario_guarantee} \mathbb{P}^N\{V(\theta^\star)>\varepsilon\}\le \sum_{i=0}^{n}\binom{N}{i}\varepsilon^i(1-\varepsilon)^{N-i}. \end{align}
Setting a confidence level $\beta\in (0,1)$, we can solve for $N$ which yields the minimum number of samples needed to guarantee that the violation probability is not larger than $\varepsilon$ with a probability of at least $1-\beta$. 

It is clear that, to guarantee small violation probabilities with a high confidence level, the amount of samples $N$ needs to be significantly larger than the dimension of the decision variable $\theta$. More precisely, as shown in \cite{Campi2018} the number of samples can be lower bounded by \begin{align}
    \label{eq:sample_bound}
    N \geq \frac{2}{\varepsilon}\left(n + \ln\left(\frac{1}{\beta}\right)\right)
\end{align}
which depends linearly on $n$. Alternatively, $N$ can be computed by applying a simple bisection which is less conservative compared to \eqref{eq:sample_bound}. 

If two scalar scenario programs are solved independently with risk $\varepsilon$ and confidence parameter $\beta$, then by the union bound the joint violation risk is at most $2\varepsilon$, and the joint confidence is at least $1-2\beta$.

Obtaining \glspl{ipm} corresponds to a specific formulation of the scenario optimization problem. The goal is to find a function $f\in \mathcal{M}_n$, where $\mathcal{M}_n:=\Span\{\phi_1,\ldots,\phi_n\}$ is spanned by $n$ linearly independent functions, while minimizing the maximum error to the training data. The optimization problem for constructing an \gls{ipm} can be formulated as a min-max problem
\begin{align}
    \label{eq:ipm_minmax}
    \min_{f\in\mathcal{M}_n} \max_{i=1,\dots,N} l(f,\delta_i),
\end{align}
where the loss function $l : \mathcal{M}_n \times (\Omega\times Y) \to \mathbb{R}$ is given as $l(f,(z,y)) = |y-f(z)|$ and $\delta_i := (z_i,y_i) \in \Omega\times Y$ is an input-output pair.
Equation \eqref{eq:ipm_minmax} can be reformulated as a scenario optimization problem
\begin{eqnarray}
    \label{eq:ipm_optimization}
    \begin{aligned}
        \min_{\alpha\in\mathbb{R}^n, \gamma \in \mathbb{R}} \quad & \gamma \\
        \text{s.t.} \quad & |y_i-f_\alpha(z_i)| \leq \gamma, \quad i=1,\ldots,N\\
        & f_\alpha = \alpha^\top \Phi,
    \end{aligned}
\end{eqnarray}
where $\alpha = [\alpha_1,\ldots,\alpha_n]^\top$ and $\Phi = [\phi_1,\ldots,\phi_n]^\top$.
Here, $\gamma$ represents the upper bound on the loss $l(f_\alpha, \delta_i)$ over all samples. Applying the previously discussed scenario optimization, the solution $\theta^\star = [\alpha^\star, \gamma^\star]^\top$ provides the parameters for a nominal model $f_{\alpha^\star}$ and a worst-case loss value $\gamma^\star$.

Traditionally, \glspl{ipm} are constructed using a fixed finite-dimensional space $\mathcal{M}_n$. This approach requires prior knowledge to select a suitable model structure. If the chosen parametric family is not rich enough, it may fail to capture the true system dynamics accurately. Non-parametric models, such as those based on \gls{rkhs}, offer a more flexible alternative. Usually, one considers an \gls{rkhs} ball of the form $B_R:=\{f\in H : \|f\|_H \leq R\}$ for some $R>0$. However, the dimension of this ball is infinite. In the next section we will discuss why kernel methods can still be used to construct \glspl{ipm} and how the resulting dimension depends on the choice of the kernel.

\section{Kolmogorov n-widths and Approximation Numbers}
\label{sec:n_widths}
As already mentioned, since $H$ is infinite-dimensional, the question is whether it is possible to find $\mathcal{M}_n\subset H$ that is finite dimensional and can approximate the set $B_R$ sufficiently well. More concretely, are there dominant dimensions in $H$ that contain the majority of the approximation capacity and, if so, how large is the worst-case approximation error when considering an $n-$dimensional subspace composed of exactly these dimensions? These questions directly lead to the so-called $n-$widths. We restrict ourselves in this work to the approximation numbers and Kolmogorov $n-$widths. 

\begin{definition}[Approximation numbers \& n-widths]
    Let $A$ and $B$ be \gls{nvs} such that $\id: A \to B$ is bounded, and let $A_E$ be a ball in $A$. The approximation number of $A_E$ in $B$ is defined as
     \begin{align}
        \label{eq:approximation_numbers}
    a_n(A_E,B) \;=\; \inf_{\substack{T \in \mathcal{L}(A,B)\\ \rank(T)\leq n}} \sup_{f\in A_E} \|f-Tf\|_B,
    \end{align}
    where the infimum is taken over all bounded linear operators $T$ that map from $A$ to $B$ with  $\mathrm{rank}(T)\leq n$.

    The Kolmogorov $n$-width of $A_E$ in $B$ is defined as
    \begin{align}
        \label{eq:K_nwidth}
    d_n(A_E,B) \;=\; \inf_{\substack{B_n \subset B \\ |B_n| = n}} \sup_{f\in A_E} \inf_{g\in B_n} \|f-g\|_B,
    \end{align}
    where the infimum is taken over all subspaces $B_n$ of $B$ with dimension $n$.
\end{definition}

Approximation numbers and the Kolmogorov $n$-width are closely related. Approximation numbers describe the worst case minimum that occurs by applying a low rank transformation. The Kolmogorov $n-$width characterizes the minimal worst-case error in $B$-norm achievable when approximating elements of a set $A_E \subset B$ by $n-$dimensional subspaces of $B$, hence, the approximation takes place in the image space $B$. Comparing both expressions, we can see that $g=Tf$ for \eqref{eq:approximation_numbers} is linear in $f$ whereas in \eqref{eq:K_nwidth} $g\in B_n$. Since the inner infimum with respect to $g$ might depend nonlinearly on $f$ in a general \gls{nvs}, the Kolmogorov $n-$widths are more flexible than the approximation numbers which implies $a_n(A_E,B) \geq d_n(A_E,B)$.

In our setting, $B$ is replaced by $L^q$ and $A_E$ by $B_R$ which denotes the \gls{rkhs} ball with radius $R$. We recall that we want to find a subspace $\mathcal{M}_n \subset H$ such that the worst case error with respect to the $L^q$ norm is minimized, see \cref{fig:$n-$width}. Hence, the usual definition of Kolmogorov $n-$widths needs to be adapted to $B_n \subset A$ in \eqref{eq:K_nwidth}. Indeed, this is a notable restriction of the approximation space such that asymptotic decays from the literature for Kolmogorov $n-$widths as in \cite{Steinwart2017} may not hold. Therefore, the more conservative approximation numbers are considered.

\begin{figure}
    \centering
    \includegraphics[width=\linewidth]{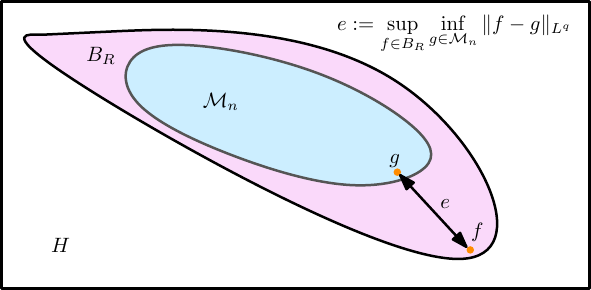}
    \caption{Error governed by the Kolmogorov $n-$widths.}
    \label{fig:$n-$width}
\end{figure}

Instead of merely establishing the existence of subspaces, we adopt a constructive perspective. To this end, we restrict attention to kernels whose native spaces are norm-equivalent to Sobolev spaces. Prominent examples include compactly supported Wendland functions \cite{Wendland2004} as well as the Mat\`ern family of kernels. Moreover, we define $X\subset\Omega$ to be a discrete set with mesh norm 
\[
   h_X := \sup_{x\in\Omega}\min_{x_i\in X}\|x-x_i\|.
\]  
First, we recall a slightly simplified version of Theorem~11.32 from \cite{Wendland2004}. We define the interpolant $P_X f$ as the orthogonal projection onto the space $\Span\{k(x,\cdot), x \in X\}$, i.e., the space spanned by kernel functions centered at the grid points of $X$.

\begin{theorem}[cf. \cite{Wendland2004}, Thm.\ 11.32, integer-order case]
\label{thm:wendland_sobolev}
Let $\Omega\subset\mathbb{R}^d$ be a bounded domain satisfying an interior cone condition. 
Fix $l\in\mathbb{N}$ and $m\in\mathbb{N}_0$, and let $1\leq p < \infty$, $p\leq q \leq \infty$. 
Assume the Sobolev embedding condition \cite{Adams2003} (this ensures that $\id: W_p^l(\Omega) \to W_q^m(\Omega)$ is continuous)
\[
   \begin{cases}
     l > m + d/p, & p>1,\\[0.3em]
     l \geq m + d, & p=1.
   \end{cases}
\]
Then, for any $f\in W^{l}_p(\Omega)$ the error satisfies
\[
  \|f - P_X f\|_{W^m_q(\Omega)} 
  \;\leq\; C\, h_X^{\,l-m - d\,(1/p - 1/q)_+}\; \|f\|_{W^{l}_p(\Omega)},
\]
where $(\alpha)_+ := \max\{\alpha,0\}$, $h_X$ sufficiently small, and $C>0$ is a constant independent of $f$ and $X$.
\end{theorem}

The definition of interior cone conditions can be found in \cite{Wendland2004} definition 3.6. Loosely speaking, it ensures that every point in $\Omega$ can be embedded into a cone, with a fixed angle and radius, which is itself fully contained in $\Omega$. The resulting angle and radius depend on the domain under consideration. Note that the interior cone condition is not very restrictive and holds for many domains, e.g., hypercubes, hyperspheres. In the following we will use Theorem~\ref{thm:wendland_sobolev} to derive bounds on the approximation numbers of kernel native spaces.

\begin{corollary}
    \label{thm:approximation_numbers_sobolev}
    Let $H$ be an \gls{rkhs} which is norm-equivalent to the Sobolev space $W_2^l(\Omega)$ for some $l\in\mathbb{N},\, l> d/2$, where $\Omega\subset\mathbb{R}^d$ is a compact domain satisfying an interior cone condition. Furthermore, let $B_R := \{f\in H: \|f\|_H \leq R\}$. Then, for $2\leq q \leq \infty$ there exists a constant $C > 0$ such that
     \[ a_n(B_R,L^q) \leq C\, n^{\,-l/d + (1/2 - 1/q)_+}.\]
\end{corollary}
\begin{proof}
    \begin{align*}
        a_n(B_R,L^q) &= \inf_{\substack{T \in \mathcal{L}(H,L^q)\\ \rank(T)\leq n}} \sup_{f\in B_R} \|f-Tf\|_{L^q}\\
        &\leq \sup_{f\in B_R} \|f-P_X f\|_{L^q}\\
        & = \sup_{f\in B_R} \|\id: W_q^0(\Omega) \to L^q(\Omega)\| \|f-P_X f\|_{W_q^0}\\
        & =  \sup_{f\in B_R} \|f-P_X f\|_{W_q^0}\\
        &\leq C_1\, h_X^{\,l - d\,(1/2 - 1/q)_+}\; R
    \end{align*}
    which follows by applying Theorem~\ref{thm:wendland_sobolev} with $p=2$ and $m=0$. The last step follows from the fact that using a quasi-uniform grid there is $C_2 > 0$ such that $h_X \leq C_2 n^{-1/d}$. Hence, we have $C = C_1\, C_2\, R$.
\end{proof}

\begin{remark}
    Since the derivation of Theorem~\ref{thm:wendland_sobolev} is constructive, one may ask whether the decay for the approximation numbers is optimal. Indeed, \cite{Steinwart2017} proved that the approximation numbers of Sobolev spaces satisfy $a_n(W_2^l(\Omega),L^q) \asymp n^{-l/d + (1/2 - 1/q)}$ for $q\in[2,\infty]$ which matches the decay rate in Corollary~\ref{thm:approximation_numbers_sobolev}. This implies asymptotic optimality of Theorem~\ref{thm:wendland_sobolev}. However, there may be fast-decaying terms that dominate for small $n$. Hence, the bound in Corollary~\ref{thm:approximation_numbers_sobolev} is not necessarily tight.
\end{remark}
\begin{remark}
    Corollary~\ref{thm:approximation_numbers_sobolev} can be extended to \glspl{rkhs} with smooth native spaces, e.g., Gaussian kernels. 
    In this case, the approximation numbers decay super-exponentially, i.e., $a_n(B_R,L^q) \leq C \exp(-c n^{1/d})$ for some $c > 0$, see \cite{Wendland2004} Theorem~11.22 with $C$ defined analogously.
\end{remark}
Note that the approximation numbers suffer from the curse of dimensionality for $q > 2$. Substituting $l = \nu + d/2$ in the statement of Corollary~\ref{thm:approximation_numbers_sobolev} we have $a_n(B_R,L^q) \leq C n^{-\nu/d + (1/2 - 1/q)_+}$. Hence, $\nu$ needs to grow at least linearly with $d$ to compensate.

\section{Application}
\label{sec:application}
Even though the bounds derived in Corollary~\ref{thm:approximation_numbers_sobolev} are not always useful in practice, they show how the order of required samples scales with the dimension of the input space as well as the smoothness of the kernel. In this section, we show that approximation numbers can be estimated using numerical data. We then apply the results to a robust system identification problem. All code is available on Zenodo \cite{lubsen_2025_code}.

\subsection{Greedy Basis Selection via the Power Function}
We consider the approximation numbers $a_n(B_R,L^\infty)$ and define a tolerance on the approximation error $\tau > 0$. Recall that the goal is to find the smallest $n$ such that
\begin{align*}
      \inf_{\substack{T \in \mathcal{L}(H,L^\infty)\\ \rank(T)\leq n}} \sup_{f\in B_R} \|f-Tf\|_{L^\infty} \leq \tau,
\end{align*}
where $B_R$ is defined as in the previous section. Let $P_Z$ be a projection operator onto the subspace spanned by $\{k(\cdot,z_i): z_i \in Z \text{ and } |Z| = n\}$. From the multi-armed bandit and safe Bayesian optimization literature, e.g., \cite{chowdhury2017,luebsen2025}, we have that 
\begin{align}
    \label{eq:min_power_fun}
    |f(x)-P_Z f(x)| \leq R \,\mathcal{P}_Z(x), \quad \forall x\in\Omega,\, \forall f \in B_R,
\end{align}
where $\mathcal{P}_Z(x) = \sqrt{k(x,x) - k(Z,x)^\top K_Z^{-1} k(Z,x)}$ is the power function, $k(Z,x) = [k(z_1,x),\dots,k(z_n,x)]^\top$ and $[K_Z]_{ij} = k(z_i,z_j)$. In the context of Bayesian optimization, the power function is also known as the posterior standard deviation of a Gaussian process. However, note that although the two functions look identical, they have completely different interpretations.

In view of \eqref{eq:min_power_fun}, to achieve a tolerance $\tau$ the goal is to find a set $Z$ such that $\inf_{Z\in\Omega^n}\sup_{x\in\Omega} \mathcal{P}_Z(x) \leq \tau/R$. This is a highly non-convex min-max problem. To simplify computations, $\Omega$ is discretized into a set $X_M=[x_1,\dots,x_M]$. Hence, we replace the continuous-domain $\Omega$ by the discrete domain $X_M$ such that
\begin{align}
    \label{eq:opt_problem}
    \inf_{Z \in \Omega^n}\max_{x\in X_M} \mathcal{P}_Z(x) \leq \tau/R.
\end{align}
This problem can be solved using a P-greedy algorithm \cite{binev2011} as given in Algorithm~\ref{algo:p_greedy}. The algorithm returns the reduced basis defined as $\mathcal{M}_n:= \text{span}\{k(\cdot,z_i) : z_i \in Z\}$ with $n = |Z|$. In contrast to \eqref{eq:opt_problem}, the tolerance $\tau$ is used as a stopping criterion instead of $n$.

\begin{algorithm}[t]
\caption{P-greedy algorithm for estimating the effective dimension}
\label{algo:p_greedy}
\KwIn{$X$, kernel $k$, tolerance $\mathrm{tol}$}
\KwOut{Basis $Z$}

$Z\leftarrow \emptyset$\; $j \leftarrow 0$\;
\While{1}{
    $z \leftarrow \arg\max_{x\in X} P_Z(x)^2$\;
    \If{$P_Z(z)^2<\mathrm{tol}$}{\textbf{break}}
    $Z \leftarrow Z\cup\{z\}$\; $j \leftarrow j+1$\;
}
\end{algorithm}

\subsection{Robust System Identification using RKHS}
\label{sec:robust_system_identification}
In this section, we show how the previous complexity analysis of \glspl{rkhs} can be used in practice jointly with the scenario approach.
We will show this in terms of a toy example, the Van der Pol oscillator, which is given by
\begin{align*}
    \dot{x}_1 &= x_2\\
    \dot{x}_2 &= (1-x_1^2)x_2 - x_1 - 2 + 2u.
\end{align*}
We define the state vector as $\x=[x_1,x_2]^\top$. After discretizing the continuous-time system, we obtain the discrete-time state-space model $\x_{k+1} = f(\x_k,u_k) + \bm{\omega}_k$, where $\bm{\omega}_k$ denotes process noise.
Using this model, we generate the dataset by randomly sampling $\x_k$ and $u_k$ from $\Omega$ and computing $\x_{k+1}$. The input and output data points are given as $\z_k = [\x_k^\top,u_k]^\top$, and $\y_k = \x_{k+1}$, respectively. After computing the required number of basis functions via Algorithm~\ref{algo:p_greedy}, $N$ is computed via bisection of the right side of \eqref{eq:scenario_guarantee}.

The goal is to design a control input $u$ that steers the system from an initial state $\x_0 = [4,0]^\top$ to the end point $\x_f = [-2,0]^\top$ while avoiding a star-shaped obstacle centered at $[0,-4]$, see \cref{fig:van_der_pol_obstacle}.
The remaining settings are specified in \cref{tab:vdp_setup}.
 
The scenario optimization problem for each state $y^{(l)}$ can be formulated as
\begin{align}
    \label{eq:scenario_opti} \nonumber
    \min_{\alpha_l \in \mathbb{R}^n, \gamma_l \in \mathbb{R}}& \gamma_l\\
    \text{s.t. }\quad & \alpha_l^\top K_Z \alpha_l \leq R^2,\\ \nonumber
    & |y^{(l)}_i - \alpha_l^\top k(Z,\z_i)| \leq \gamma_l, \quad \forall\,i=1,\dots,N.
\end{align}
Taking the union bound over both states yields a total violation risk of at most $
0.05$ and confidence of at least $1-2\cdot 10^{-6}$. In \cref{tab:vdp_setup}, we can see that the resulting radii $\gamma_1$ and $\gamma_2$ are close to the third standard deviation $3\,\sigma_\mathrm{noise}$, which includes 99.7\% of the probability mass.

\begin{figure}
    \centering
    \input{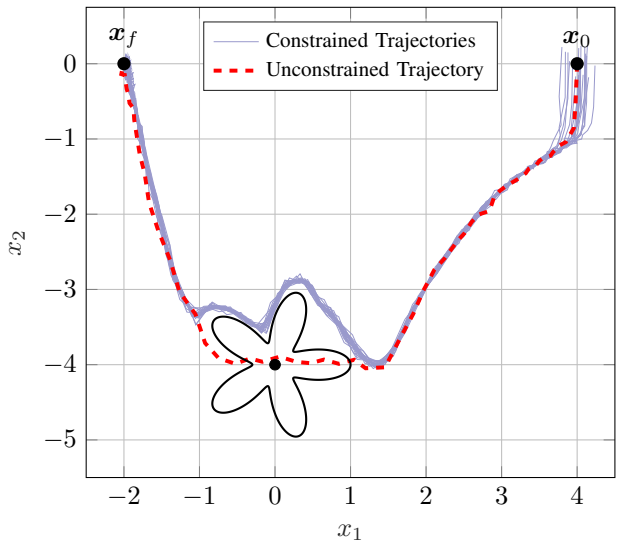}
    \caption{Obstacle avoidance example using a robust \gls{rkhs} model of the Van der Pol oscillator. The red line indicates the optimal trajectory without obstacle avoidance. The blue lines are optimal trajectories with obstacle avoidance and sampled initial states around $x_0$.}
    \label{fig:van_der_pol_obstacle} 
\end{figure}

\cref{fig:van_der_pol_obstacle} shows the results of the robust system identification and control. The robust model predictive controller is implemented using Acados~\cite{Verschueren2021}, which operates with a prediction horizon of 30 steps and solves each optimization in approximately (0.4 - 0.7)\,ms on a standard notebook with an Intel® Core™ i7-8565U CPU. The initial state for each of the 20 blue trajectories is sampled around $\x_0$. Moreover, the tube radii are propagated over the horizon by recursively evaluating the model at the four corners in the state space, i.e., $\y_k + [\pm\gamma_1,\pm\gamma_2]^\top$. Of course, this does not rigorously ensure that the whole MPC rollout is safe. As demonstrated in \cref{fig:van_der_pol_obstacle}, the controller successfully steers the system to the end point while respecting all constraints.

\begin{table}[t]
\centering
\caption{Setup for the Van der Pol example.}
\label{tab:vdp_setup}
\renewcommand{\arraystretch}{1.15}
\begin{tabular}{p{4.cm}p{3.cm}}
\toprule
\textbf{Quantity} & \textbf{Value} \\
\midrule
Sampling domain & $\Omega :=[-5,5]^3$ \\
Kernel & Mat\'ern with $\nu = 5/2$ \\
Noise & $\bm{\omega}_k = \mathcal{N}(0,\sigma_\mathrm{noise}^2 I_2)$\\
Noise standard deviation & $\sigma_\mathrm{noise}=0.02$\\
Tube sizes & $\gamma_1 \approx 0.057$, $\gamma_2\approx 0.068$\\
Sampling time & $T_s = 0.1$\\
Tolerance & $\tau = 0.1$ \\
RKHS norm bound & $R = 350$ \\
Risk & $\epsilon = 0.025$ \\
Confidence & $\beta = 10^{-6}$ \\
Number of scenarios & $N = 4200$ \\
Dataset & $\Delta = \{(\z_i,\y_i)\}_{i=1}^N$ \\
Number of basis functions & $n = 60$ \\
Optimization variables & $\alpha_l \in \mathbb{R}^n,\ \gamma_l \in \mathbb{R}$ \\
\bottomrule
\end{tabular}
\end{table}

\subsection{Discussion}
In the introduction, we claimed that this approach, unlike methods based on uncertainty bounds for support vector regression, does not rely on strong assumptions about the unknown system, such as its \gls{rkhs} norm. Indeed, the norm bound $R$ in \eqref{eq:scenario_optimization} only restricts the hypothesis set to an \gls{rkhs} ball. It does not claim to contain the true system. 
The radius $\gamma_l$ obtained from \eqref{eq:scenario_optimization} depends on $R$, e.g., if $R$ is too small, the resulting radius $\gamma_l$ will be large because the function class is not allowed to vary too much. Nevertheless, the theoretical guarantees for the constructed tube will hold.

In contrast to the presented approach where we have defined the optimization variables a priori, the a posteriori scenario framework \cite{campi2018} can be used to bound the violation probability using the number of active support constraints that are observed after solving the optimization problem. Combined with an $l_1$-regularization term in the optimization problem, this can further reduce the number of active support constraints and hence the violation probability. However, as shown in \cref{thm:approximation_numbers_sobolev}, we expect that the number of active support constraints scales with the dimension of the input space, which implies that huge Gram matrices need to be stored and evaluated. This is a major bottleneck and quickly leads to memory issues. This is exactly the reason why the P-greedy algorithm is applied on an extra training set, in order to compress the data. Since the current theoretical guarantees do not allow the same training set to be used twice, there are two ways to obtain an additional dataset. The first option is to sample from some space that contains the system's manifold, e.g. a hypercube. The second option is to sample directly from the system’s trajectories, i.e., by running the system and collecting data. Of course, the first approach is easier to implement, but it may lead to a larger number of support vectors due to the fact that the samples are not restricted to the system's manifold. The second approach provides a better estimate, since the effective dimension is lower, but it requires running the system and collecting data, which may be costly.
To improve data efficiency, we plan to investigate compression-based learning approaches \cite{campi2024,rocchetta2024} in future work. These approaches can reduce the number of support constraints beforehand.

Another major bottleneck in dynamical-system applications is the i.i.d. requirement on the data. In practice, of course, this is barely satisfied. Usually, one assumes that the underlying dynamical system mixes over time, which permits subsampling from trajectories.
\section{Conclusion}
\label{sec:conclusion}
In this paper, we presented a new method for nonlinear system identification with finite-sample one-step violation guarantees within an \gls{rkhs} framework. Our approach successfully bridges the scenario approach with approximation theory to provide these guarantees without requiring a correct a priori upper bound on the true system's \gls{rkhs} norm. We have shown that the required amount of data to achieve a desired approximation accuracy increases exponentially with the input dimension. Moreover, the i.i.d. assumption on the data is not satisfied in practice, which further complicates the application of the scenario approach. Future work will focus on addressing these limitations by exploring compression-based methods to improve memory efficiency and by investigating extensions/alternatives to the scenario approach that can handle dependent data.



\printbibliography

@book{Adams2003,
  title = {Sobolev Spaces},
  author = {Adams, Robert A. and Fournier, John J. F.},
  date = {2003},
  series = {Pure and Applied Mathematics},
  edition = {2nd ed},
  number = {v. 140},
  publisher = {Academic Press},
  location = {Amsterdam Boston},
  isbn = {978-0-12-044143-3},
  langid = {english}
}

@article{Bradford2020a,
  title = {Stochastic Data-Driven Model Predictive Control Using Gaussian Processes},
  author = {Bradford, Eric and Imsland, Lars and Zhang, Dongda and Del Rio Chanona, Ehecatl Antonio},
  date = {2020-08},
  journaltitle = {Computers \& Chemical Engineering},
  volume = {139},
  % pages = {106844},
  % issn = {00981354},
  doi = {10.1016/j.compchemeng.2020.106844},
  % url = {https://linkinghub.elsevier.com/retrieve/pii/S0098135419313080},
  langid = {english}
}

@article{pillonetto2018,
title = {System Identification Using Kernel-Based Regularization: New Insights on Stability and Consistency Issues},
journal = {Automatica},
volume = {93},
pages = {321-332},
year = {2018},
% issn = {0005-1098},
doi = {10.1016/j.automatica.2018.03.065},
% url = {https://www.sciencedirect.com/science/article/pii/S0005109818301602},
author = {Gianluigi Pillonetto},
% keywords = {Learning from examples, System identification, Reproducing kernel Hilbert spaces of dynamic systems, Kernel-based regularization, BIBO stability, Regularization networks, Generalization and consistency},
}

@ARTICLE{yin2023,
  author={Yin, Mingzhou and Smith, Roy S.},
  journal={IEEE Control Systems Letters}, 
  title={Error Bounds for Kernel-Based Linear System Identification With Unknown Hyperparameters}, 
  year={2023},
  volume={7},
  number={},
  % pages={2491-2496},
  % keywords={Kernel;Linear systems;Stochastic processes;System identification;Standards;Probabilistic logic;Numerical models;Identification;statistical learning;machine learning;uncertain systems},
  doi={10.1109/LCSYS.2023.3287305}
  }

@article{suykens2010,
author = {Johan A.K. Suykens and Carlos Alzate and Kristiaan Pelckmans},
title = {{Primal and Dual Model Representations in Kernel-Based Learning}},
volume = {4},
journal = {Statistics Surveys},
publisher = {Amer. Statist. Assoc., the Bernoulli Soc., the Inst. Math. Statist., and the Statist. Soc. Canada},
% pages = {148 -- 183},
year = {2010},
doi = {10.1214/09-SS052},
% URL = {https://doi.org/10.1214/09-SS052}
}

@article{khosravi2023,
title = {The Existence and Uniqueness of Solutions for Kernel-Based System Identification},
journal = {Automatica},
volume = {148},
pages = {110728},
year = {2023},
% issn = {0005-1098},
doi = {10.1016/j.automatica.2022.110728},
% url = {https://www.sciencedirect.com/science/article/pii/S0005109822005945},
author = {Mohammad Khosravi and Roy S. Smith},
% keywords = {System identification, Kernel-based methods, Existence and uniqueness of solution, Integrable kernels},
}

@book{kocijan2016,
  author = {Juš Kocijan},
  editor = {},
  publisher = {Springer Cham},
  title = {Modelling and Control of Dynamic Systems Using Gaussian Process Models},
  year = {2016},
  isbn = {978-3-319-21021-6},
  doi = {10.1007/978-3-319-21021-6},
}

@book{Campi2018,
  title = {Introduction to the {{Scenario Approach}}},
  author = {Campi, Marco C. and Garatti, Simone},
  date = {2018-11-05},
  publisher = {{Society for Industrial and Applied Mathematics}},
  location = {Philadelphia, PA},
  doi = {10.1137/1.9781611975444},
  % url = {https://epubs.siam.org/doi/book/10.1137/1.9781611975444},
  isbn = {978-1-61197-543-7 978-1-61197-544-4},
  langid = {english}
}

@inproceedings{Kocijan2004,
  title = {Gaussian Process Model Based Predictive Control},
  booktitle = {Proceedings of the 2004 {{American Control Conference}}},
  author = {Kocijan, J. and Murray-Smith, R. and Rasmussen, C.E. and Girard, A.},
  date = {2004},
  % pages = {2214-2219 vol.3},
  publisher = {IEEE},
  location = {Boston, MA, USA},
  doi = {10.23919/ACC.2004.1383790},
  % url = {https://ieeexplore.ieee.org/document/1383790/},
  eventtitle = {Proceedings of the 2004 {{American Control Conference}}},
  isbn = {978-0-7803-8335-7},
  langid = {english}
}

@online{Kohne2024,
  title = {${L^\infty}$-Error Bounds for Approximations of the {{Koopman}} Operator by Kernel Extended Dynamic Mode Decomposition},
  author = {Köhne, Frederik and Philipp, Friedrich M. and Schaller, Manuel and Schiela, Anton and Worthmann, Karl},
  date = {2024-07-04},
  eprint = {2403.18809},
  eprinttype = {arXiv},
  eprintclass = {math},
  doi = {10.48550/arXiv.2403.18809},
  % url = {http://arxiv.org/abs/2403.18809},
  langid = {english},
  pubstate = {prepublished},
  % keywords = {Computer Science - Numerical Analysis,Mathematics - Dynamical Systems,Mathematics - Numerical Analysis}
}

@online{Scampicchio2025,
  title = {Gaussian Processes for Dynamics Learning in Model Predictive Control},
  author = {Scampicchio, Anna and Arcari, Elena and Lahr, Amon and Zeilinger, Melanie N.},
  date = {2025-02-04},
  eprint = {2502.02310},
  eprinttype = {arXiv},
  eprintclass = {eess},
  doi = {10.48550/arXiv.2502.02310},
  % url = {http://arxiv.org/abs/2502.02310},
  langid = {english},
  pubstate = {prepublished},
  % keywords = {Computer Science - Systems and Control,Electrical Engineering and Systems Science - Systems and Control}
}

@article{binev2011,
author = {Binev, Peter and Cohen, Albert and Dahmen, Wolfgang and DeVore, Ronald and Petrova, Guergana and Wojtaszczyk, Przemyslaw},
title = {Convergence Rates for Greedy Algorithms in Reduced Basis Methods},
journal = {SIAM Journal on Mathematical Analysis},
volume = {43},
number = {3},
% pages = {1457-1472},
year = {2011},
doi = {10.1137/100795772},

% URL = { 
%         https://doi.org/10.1137/100795772
% },
% eprint = { 
    
%         https://doi.org/10.1137/100795772},
}

@misc{campi2024,
      title={Compression, Generalization and Learning}, 
      author={Marco C. Campi and Simone Garatti},
      year={2024},
      eprint={2301.12767},
      archivePrefix={arXiv},
      primaryClass={cs.LG},
      url={https://arxiv.org/abs/2301.12767}, 
}

@ARTICLE{rocchetta2024,
  author={Rocchetta, Roberto and Mey, Alexander and Oliehoek, Frans A.},
  journal={IEEE Transactions on Neural Networks and Learning Systems}, 
  title={A Survey on Scenario Theory, Complexity, and Compression-Based Learning and Generalization}, 
  year={2024},
  volume={35},
  number={12},
  pages={16985-16999},
  % keywords={Data models;Complexity theory;Numerical models;Statistical learning;Decision making;Support vector machine classification;Picture archiving and communication systems;Agnostic learning;compression;generalization theory;probably approximately correct (PAC);scenario optimization;support vector classifiers},
  doi={10.1109/TNNLS.2023.3308828}}

@misc{luebsen2025,
      title={An Analysis of Safety Guarantees in Multi-Task Bayesian Optimization}, 
      author={Jannis Luebsen and Annika Eichler},
      year={2025},
      eprint={2503.08555},
      archivePrefix={arXiv},
      primaryClass={cs.LG},
      url={https://arxiv.org/abs/2503.08555}, 
}

@misc{chowdhury2017,
      title={On Kernelized Multi-armed Bandits}, 
      author={Sayak Ray Chowdhury and Aditya Gopalan},
      year={2017},
      eprint={1704.00445},
      archivePrefix={arXiv},
      primaryClass={cs.LG},
      url={https://arxiv.org/abs/1704.00445}, 
}

@article{Steinwart2017,
  title = {A Short Note on the Comparison of Interpolation Widths, Entropy Numbers, and {{Kolmogorov}} Widths},
  author = {Steinwart, Ingo},
  date = {2017-03},
  journaltitle = {Journal of Approximation Theory},
  volume = {215},
  % pages = {13--27},
  % issn = {00219045},
  doi = {10.1016/j.jat.2016.11.006},
  % url = {https://linkinghub.elsevier.com/retrieve/pii/S0021904516301071},
  langid = {english}
}

@article{Verschueren2021,
  title = {Acados – a Modular Open-Source Framework for Fast Embedded Optimal Control},
  author = {Verschueren, Robin and Frison, Gianluca and Kouzoupis, Dimitris and Frey, Jonathan and family=Duijkeren, given=Niels, prefix=van, useprefix=true and Zanelli, Andrea and Novoselnik, Branimir and Albin, Thivaharan and Quirynen, Rien and Diehl, Moritz},
  date = {2021},
  journaltitle = {Mathematical Programming Computation}
}

@book{Wendland2004,
  title = {Scattered {{Data Approximation}}},
  author = {Wendland, Holger},
  date = {2004},
  series = {Cambridge {{Monographs}} on {{Applied}} and {{Computational Mathematics}}},
  isbn = {978-0-521-84335-5},
  langid = {english}
}

@article{Williams2015b,
  title = {A {{Data}}–{{Driven Approximation}} of the {{Koopman Operator}}: {{Extending Dynamic Mode Decomposition}}},
  author = {Williams, Matthew O. and Kevrekidis, Ioannis G. and Rowley, Clarence W.},
  date = {2015-12-01},
  journaltitle = {Journal of Nonlinear Science},
  volume = {25},
  number = {6},
  % pages = {1307--1346},
  % issn = {1432-1467},
  doi = {10.1007/s00332-015-9258-5},
  % url = {https://doi.org/10.1007/s00332-015-9258-5}
}

@software{lubsen_2025_code,
  author       = {Luebsen, Jannis and
                  Eichler, Annika},
  title        = {code-Robust-Nonlinear-System-
                   Identification-using-Reproducing-Kernel-Hilbert-
                   Spaces: Code for Final Submission },
  month        = apr,
  year         = 2026,
  publisher    = {Zenodo},
  version      = {v1.0},
  doi          = {10.5281/zenodo.19438697},
}

\end{document}